\documentclass[conference,letterpaper]{IEEEtran}
\IEEEoverridecommandlockouts
\usepackage{cite}
\usepackage{url}
\usepackage{ifthen}
\usepackage{algorithm, algorithmicx, algpseudocode}
\usepackage{graphicx} 
\usepackage{textcomp}
\usepackage{booktabs}
\usepackage{makecell} 
\usepackage[cmex10]{amsmath}
\usepackage{amssymb,amsfonts}
\usepackage{pifont}
\usepackage[left=1.62cm,right=1.62cm,top=1.85cm]{geometry}
\usepackage{amsthm}
\usepackage{mathrsfs}
\def\BibTeX{{\rm B\kern-.05em{\sc i\kern-.025em b}\kern-.08em T\kern-.1667em\lower.7ex\hbox{E}\kern-.125emX}}
\usepackage{tikz}
\usetikzlibrary{automata,positioning,calc}
\usetikzlibrary{intersections}
\usetikzlibrary{decorations.pathreplacing,angles,quotes}

\usepackage{bm}
\usepackage{amscd}

\algnewcommand{\Initialize}[1]{%
  \State \textbf{Initialization:}
  \Statex \hspace*{\algorithmicindent}\parbox[t]{0.8\linewidth}{\raggedright #1}
}
\makeatletter

\theoremstyle{definition}

\newtheorem{definition}{Definition}

\newtheorem{prop}{Proposition}
\newtheorem{cor}{Corollary}
\newtheorem{remark}{Remark}

\newcommand{\abs}[1]{\left\lvert#1\right\rvert}

\def\vE{\mathbb E}

\font\b=cmr10 scaled\magstep4

\def\bigzerou{\smash{\lower1.7ex\hbox{\b 0}}}
\def\bigzerou{\smash{\lower1.7ex\hbox{\b 0}}}
 
\begin{document}

\title{A New Algorithm for Computing $\alpha$-Capacity 
\thanks{This work was supported by JSPS KAKENHI Grant Number JP23K16886.}
}

\author{
\IEEEauthorblockN{Akira Kamatsuka}
\IEEEauthorblockA{Shonan Institute of Technology \\ 
Email: \text{kamatsuka@info.shonan-it.ac.jp}
 }
\and
\IEEEauthorblockN{Koki Kazama}
\IEEEauthorblockA{Shonan Institute of Technology \\ 
Email: \text{kazama@info.shonan-it.ac.jp}
 }
\and
\IEEEauthorblockN{Takahiro Yoshida}
\IEEEauthorblockA{Nihon University \\ 
Email: \text{yoshida.takahiro@nihon-u.ac.jp}
 } 
}

\maketitle

\begin{abstract}
The problem of computing $\alpha$-capacity for $\alpha>1$ is equivalent to that of computing the correct decoding exponent. 
Various algorithms for computing them have been proposed, such as Arimoto and Jitsumatsu--Oohama algorithm.
In this study, we propose a novel alternating optimization algorithm for computing the $\alpha$-capacity for $\alpha>1$ based on a variational characterization of the Augustin--Csisz{\'a}r mutual information. 
A comparison of the convergence performance of these algorithms is demonstrated through numerical examples.
\end{abstract}
\section{Introduction}
In the problem of channel coding for a discrete memoryless channel $p_{Y\mid X}$, channel capacity $C:=\max_{p_{X}} I(X; Y)$ \cite{shannon} and the Gallager function $E_{0}(\rho, p_{X}), \rho\in (-1, 1)$ \cite{Gallager:1968:ITR:578869} 
play a central role for analyzing the performance of codes $\mathcal{C}$, where $I(X; Y)$ is the Shannon mutual information (MI) and $p_{X}$ is an arbitrary input distribution on a finite alphabet $\mathcal{X}$. 
Shannon \cite{shannon} showed that the supremum of all achievable coding rates $R$ can be characterized by channel capacity $C$. 
Meanwhile, Gallager \cite{Gallager:1968:ITR:578869} showed the existence of an $N$-length block code $\mathcal{C}$ with a rate $R$ such that its decoding error probability $P_{e}(\mathcal{C})$ is upper bounded as $P_{e}(\mathcal{C}) \leq \exp\{-N \cdot E(R)\}$, where $E(R):=\max_{\rho\in [0, 1]} \{-\rho R + \max_{p_{X}}E_{0}(\rho, p_{X})\}$ is referred to as the \textit{error exponent}.
On the other hand, Arimoto \cite{1055007} showed that for any $N$-length code $\mathcal{C}$ with a rate $R$, 
$P_{e}(\mathcal{C})$ is lower bounded as 
$P_{e}(\mathcal{C}) \geq 1-\exp\{-N \cdot G_{\text{AR}}(R)\}$, where $G_{\text{AR}}(R):=\max_{\rho\in (-1, 0)} \{-\rho R + \min_{p_{X}} E_{0}(\rho, p_{X})\}$ is the \textit{correct decoding exponent}.

Various extensions of Shannon MI and channel capacity have been proposed and used in the analysis of other problems. 
Well-known extensions are \textit{$\alpha$-mutual information} ($\alpha$-MI \cite{7308959}) $I_{\alpha}^{(\cdot)}(X; Y)$ and \textit{$\alpha$-capacity} $C_{\alpha}^{(\cdot)}:=\max_{p_{X}}I_{\alpha}^{(\cdot)}(X; Y), \alpha\in (0, 1)\cup (1, \infty)$, some of which 
have close connection to the Gallager function $E_{0}(\rho, p_{X})$. 
The class of $\alpha$-MI includes the Sibson MI $I_{\alpha}^{\text{S}}(X; Y)$ \cite{Sibson1969InformationR}, 
Arimoto MI $I_{\alpha}^{\text{A}}(X; Y)$ \cite{arimoto1977}, Augustin--Csisz\'{a}r MI $I_{\alpha}^{\text{C}}(X; Y)$ \cite{370121}, and Lapidoth--Pfister MI $I_{\alpha}^{\text{LP}}(X; Y)$ \cite{e21080778}, \cite{8231191}. 
Their theoretical properties have been developed through many studies in literature \cite{7282554}, \cite{e21100969}, \cite{e22050526}, \cite{Nakiboglu:2019aa,8423117}, \cite{9611513},\cite{9834875}, \cite{8849809},  \cite{8804205}, \cite{e23060702}, 
including their operational meaning in problems of hypothesis testing \cite{e23020199}, \cite{6034266}, \cite{8007073} and privacy-guaranteed data-publishing \cite{8804205}.

It is known that the Sibson, Arimoto, and Augustin--Csisz{\'a}r capacity are all equivalent\footnote{It is worth mentioning that the Lapidoth--Pfister capacity $C_{\alpha}^{\text{LP}}$ is also equivalent to these capacities for $\alpha \in (1, \infty)$ \cite[Thm 4]{e22050526}, \cite[p.4]{e21100969}.} for $\alpha\in (0, 1)\cup (1, \infty)$, i.e., 
$C_{\alpha}^{\text{S}}=C_{\alpha}^{\text{A}}=C_{\alpha}^{\text{C}}$ \cite{arimoto1977},\cite{370121},\cite{e22050526}.
Combining this equivalence with the fact that Sibson MI is represented using the Gallager function $I_{\alpha}^{\text{S}}(X; Y) = \frac{\alpha}{1-\alpha} E_{0}(1/\alpha-1, p_{X})$ \cite{7308959} shows that 
maximizing $E_{0}(\rho, p_{X})$ with respect to $p_{X}$ for a fixed $\rho:=1/\alpha-1\in [0, 1]$ is equivalent to computing $\alpha$-capacity for $\alpha \in [1/2, 1]$, while  
minimizing $E_{0}(\rho, p_{X})$ with respect to $p_{X}$ for a fixed $\rho:=1/\alpha-1\in (-1, 0)$ is equivalent to computing $\alpha$-capacity for $\alpha \in (1, \infty)$.

Various algorithms have been proposed to calculate $\alpha$-capacity and exponents.
Arimoto \cite{arimoto1977,1055640} proposed an alternating optimization (AO) algorithm to calculate the Sibson capacity and exponents for a fixed $\rho$ by extending the well-known AO algorithm to compute the channel capacity proposed by Arimoto and Blahut \cite{1054753}, \cite{1054855}. 
Later, Arimoto \cite{BN01990060en} proposed an AO algorithm to directly compute Arimoto capacity. 
The AO algorithms were derived from variational characterizations\footnote{All functionals in the variational characterizations in this section are defined formally in Section \ref{sec:preliminaries} and \ref{sec:main_result}.} of the Sibson MI $I_{\alpha}^{\text{S}}(X; Y)=\max_{r_{X\mid Y}} F_{\alpha}^{\text{S1}}(p_{X}, r_{X\mid Y})$ and Arimoto MI $I_{\alpha}^{\text{A}}(X; Y)=\max_{r_{X\mid Y}}F_{\alpha}^{\text{A1}}(p_{X}, r_{X\mid Y})$. 
Kamatsuka \textit{et al}. \cite{kamatsuka2024new} proposed algorithms for computing these capacities based on recently developed variational characterizations of the Sibson MI $I_{\alpha}^{\text{S}}(X; Y)=\max_{r_{X\mid Y}}F_{\alpha}^{\text{S2}}(p_{X}, r_{X\mid Y})$ and Arimoto MI $I_{\alpha}^{\text{A}}(X; Y)=\max_{r_{X\mid Y}}F_{\alpha}^{\text{A2}}(p_{X}, r_{X\mid Y})$. They also showed that, under appropriate conditions imposed on initial distributions of the algorithms, these are equivalent to those proposed by Arimoto. 
Jitsumatsu and Oohama \cite{8889422} proposed another AO algorithm for computing $\min_{p_{X}}E_{0}(\rho, p_{X})$ for a fixed $\rho\in (-1, 0)$ based on a variational characterization of $\min_{p_{X}}E_{0}(\rho, p_{X}) = \min_{q_{X, Y}}\min_{\tilde{q}_{X, Y}} F_{\rho}^{\text{JO}}(q_{X, Y}, \tilde{q}_{X, Y})$. 
By using this characterization, another variational characterization of the Sibson capacity $C_{\alpha}^{\text{S}} = \tilde{F}_{\alpha}^{\text{JO}}(q_{X, Y}, \tilde{q}_{X, Y})$ can be obtained. Thus, another AO algorithm for computing $C_{\alpha}^{\text{S}}$ can be obtained from this characterization.

In this study, we propose a novel algorithm for computing $\alpha$-capacity for $\alpha \in (1, \infty)$ based on a variational characterization of the Augustin--Csisz{\'a}r MI $I_{\alpha}^{\text{C}}(X; Y) = \max_{\tilde{q}_{Y\mid X}}\max_{r_{X\mid Y}} \tilde{F}_{\alpha}^{\text{C}}(p_{X}, \tilde{q}_{Y\mid X}, r_{X\mid Y})$, 
which was recently derived by Kamatsuka \textit{et al}.\cite{kamatsuka2024algorithms}. 
The main contributions of this study are summarized as follows:
\begin{itemize}
\item We propose an AO algorithm for directly computing the Augustin--Csisz{\'a}r capacity $C_{\alpha}^{\text{C}}$ for $\alpha\in (1, \infty)$ by solving the triple maximization problem: $C_{\alpha}^{\text{C}} = \max_{p_{X}}\max_{\tilde{q}_{Y\mid X}}\max_{r_{X\mid Y}}\tilde{F}_{\alpha}^{\text{C}}(p_{X}, \tilde{q}_{Y\mid X}, r_{X\mid Y})$. (Algorithm \ref{alg:Csiszar})
\item We compare the convergence performances of all algorithms for computing $\alpha$-capacity proposed so far through numerical examples (Section \ref{sec:numerical_example}). 
\end{itemize}
Figure \ref{fig:algo} illustrates the relationship between the previously proposed algorithms and the proposed algorithm for calculating $\alpha$-capacity for $\alpha \in (1, \infty)$.


\begin{figure}[htbp]
\centering
\includegraphics[width=9.1cm, clip]{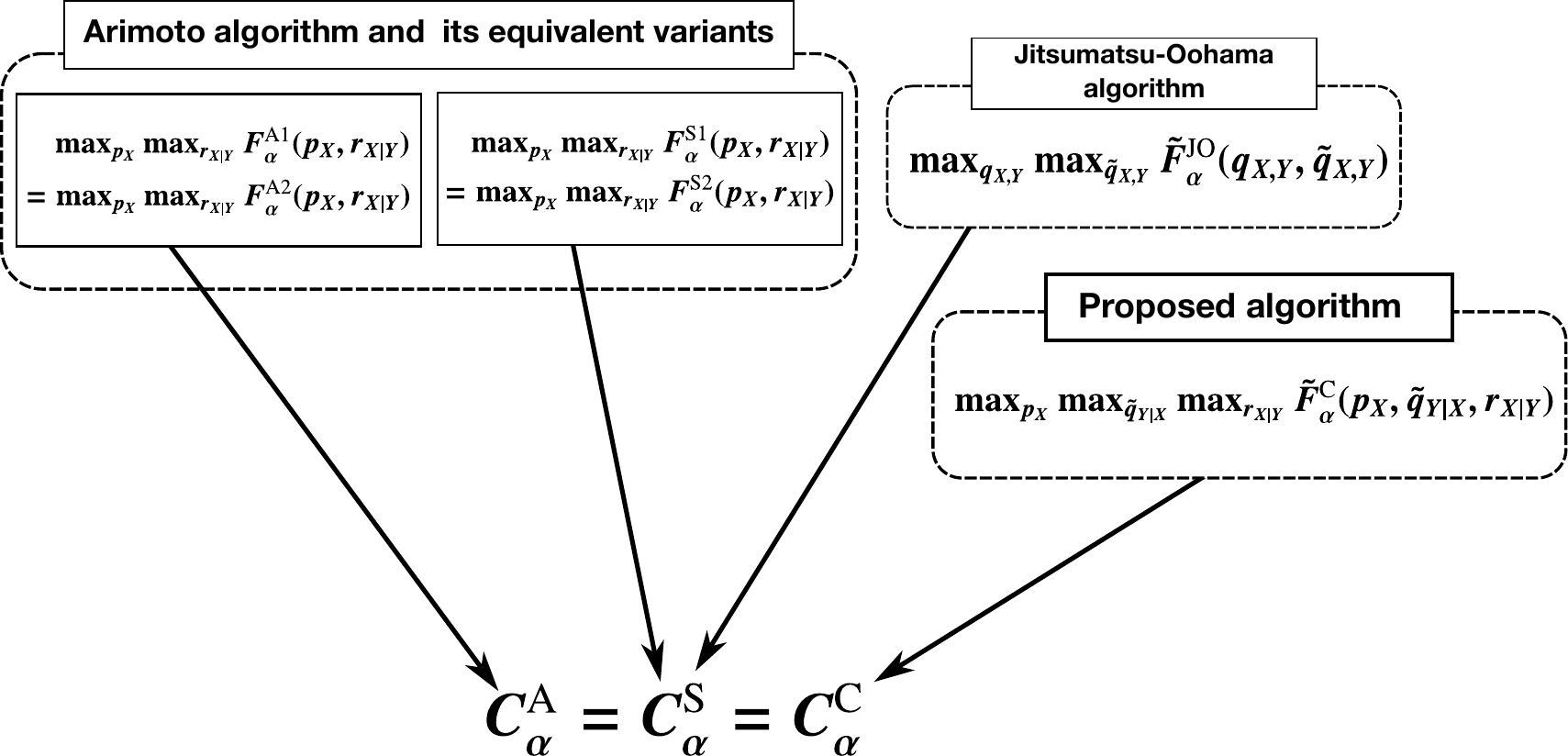}
\caption{Relationship between the algorithms for computing $\alpha$-capacity}
\label{fig:algo}
\end{figure}

\section{Preliminaries}\label{sec:preliminaries}
Let $X$ and $Y$ be random variables on finite alphabets $\mathcal{X}$ and $\mathcal{Y}$, drawn according to the joint distribution $p_{X, Y} = p_{X}p_{Y\mid X}$.
Let $p_{Y}$ be a marginal distribution of $Y$, $H(X)=H(p_{X}):=-\sum_{x}p_{X}(x)\log p_{X}(x)$ be the Shannon entropy, $H(X | Y):=-\sum_{x,y}p_{X}(x)p_{Y\mid X}(y | x)\log p_{X\mid Y}(x | y)$ be the conditional entropy, 
$I(X; Y) = I(p_{X}, p_{Y\mid X}):= H(X) - H(X | Y)$ be the mutual information, 
and $D(p_{X}||q_{X}):=\sum_{x} p_{X}(x)\log \frac{p_{X}(x)}{q_{X}(x)}$ be the Kullback--Leibler divergence between $p_{X}$ and $q_{X}$. 
We denote the expectation of $f(X)$ as $\vE_{X}^{p_{X}}[f(X)]:=\sum_{x}f(x)p_{X}(x)$. 
Throughout this paper, we use $\log$ to represent the natural logarithm.
Here, we review $\alpha$-MI, $\alpha$-capacity and their calculation algorithms.

\subsection{$\alpha$-MI and $\alpha$-Capacity}

\begin{definition} 
Let $\alpha\in (0, 1)\cup (1, \infty)$. 
The \textit{Sibson MI of order $\alpha$} \cite[Def 2.1]{Sibson1969InformationR},\cite[Eq.(13)]{370121}, the \textit{Arimoto MI of order $\alpha$} \cite[Eq.(15)]{arimoto1977}, and the \textit{Augustin--Csisz{\' a}r MI of order $\alpha$} \cite{augusting_phd_thesis},\cite[Eq. (9)]{370121} are defined as follows:
\begin{align}
I_{\alpha}^{\text{S}} (X; Y) &:= \min_{q_{Y}} D_{\alpha} (p_{X}p_{Y\mid X} || p_{X}q_{Y}) \label{eq:def_Sibson_MI}\\ 
&=  \frac{\alpha}{1-\alpha} E_{0} \left( \frac{1}{\alpha}-1, p_{X} \right), \label{eq:closed_form_Sibson_MI} \\ 
I_{\alpha}^{\text{A}}(X; Y) &:= H_{\alpha}(X) - H_{\alpha}^{\text{A}}(X\mid Y)\\
&= \frac{\alpha}{1-\alpha} E_{0}\left( \frac{1}{\alpha}-1, p_{X_{\alpha}} \right) \label{eq:closed_form_Arimoto_MI}, \\ 
I_{\alpha}^{\text{C}}(X; Y) &:= \min_{q_{Y}} \vE_{X}^{p_{X}}\left[D_{\alpha}(p_{Y\mid X}(\cdot\mid X) || q_{Y})\right],  \label{eq:Csiszar_MI} 
\end{align}
where $D_{\alpha}(p_{X} || q_{X}) := \frac{1}{\alpha-1}\log \sum_{x} p_{X}(x)^{\alpha}q_{X}(x)^{1-\alpha}$ is the R{\' e}nyi divergence between $p_{X}$ and $q_{X}$ of order $\alpha$, 
$H_{\alpha}(X) := \frac{1}{1-\alpha} \log \sum_{x} p_{X}(x)^{\alpha}$ is the R{\' e}nyi entropy of order $\alpha$ \cite{renyi1961measures}, 
$H_{\alpha}^{\text{A}}(X | Y):= \frac{\alpha}{1-\alpha}\log\sum_{y} \left( \sum_{x}p_{X}(x)^{\alpha}p_{Y\mid X}(y | x)^{\alpha} \right)^{\frac{1}{\alpha}}$ is the Arimoto conditional entropy of order $\alpha$ \cite{arimoto1977}, 
$E_{0}(\rho, p_{X}):= -\log \sum_{y}\left( \sum_{x}p_{X}(x)p_{Y\mid X}(y | x)^{\frac{1}{1+\rho}} \right)^{1+\rho}$ is the Gallager error exponent function \cite{Gallager:1968:ITR:578869},
and 
\begin{align}
p_{X_{\alpha}}(x) := \frac{p_{X}(x)^{\alpha}}{\sum_{x}p_{X}(x)^{\alpha}} \label{eq:alpha_tilted_dist}
\end{align}
is the $\alpha$-tilted distribution of $p_{X}$ \cite{8804205}. 
\end{definition}

\begin{definition}[$\alpha$-capacity]
Let $\alpha\in (0, 1) \cup (1, \infty)$ and $I_{\alpha}^{(\cdot)}(X; Y)$ be $\alpha$-MI. Then, the \textit{$\alpha$-capacity} is defined as 
\begin{align}
C_{\alpha}^{(\cdot)} &:= \max_{p_{X}} I_{\alpha}^{(\cdot)}(X; Y), 
\end{align}
where the maximum is taken over all possible probability distributions on $\mathcal{X}$.
\end{definition}
\begin{remark}
The values of $\alpha$-MI $I_{\alpha}^{(\cdot)}(X; Y)$ and $\alpha$-capacity $C_{\alpha}^{(\cdot)}$ are extended by continuity to $\alpha=1$ and $\alpha=\infty$. 
In particular, for $\alpha=1$, $\alpha$-MI and $\alpha$-capacity reduce to the Shannon MI $I(X; Y)$ and the channel capacity $C$, respectively. 
\end{remark}

It is known that the Sibson, Arimoto, and Augustin--Csisz{\'a}r capacity are all equivalent.
\begin{prop}[\text{\cite[Lemma 1]{arimoto1977},\cite[Prop 1]{370121},\cite[Thm 4]{e22050526}}] 
Let $\alpha\in (0, 1)\cup (1, \infty)$. Then, 
\begin{align}
C_{\alpha}^{\text{S}}=C_{\alpha}^{\text{A}}=C_{\alpha}^{\text{C}}. 
\end{align}
\end{prop}

\subsection{Arimoto Algorithm for Computing the Sibson and Arimoto Capacity}
In \cite{arimoto1977,BN01990060en,1055640}, Arimoto derived alternating optimization (AO) algorithms for computing $C_{\alpha}^{\text{S}}$ and $C_{\alpha}^{\text{A}}$. 
Later, Kamatsuka \textit{et al.} proposed similar AO algorithms for computing them in \cite{kamatsuka2024new}. 
These algorithms are based on the following variational characterization of $I_{\alpha}^{\text{S}}(X; Y)$ and $I_{\alpha}^{\text{A}}(X; Y)$:  

\begin{prop}[\text{\cite[Thm 1]{1055640},\cite[Eq. (7.103)]{BN01990060en},\cite[Thm 1]{kamatsuka2024new}}] \label{prop:vc_Sibson_Arimoto_MI}
Let $\alpha\in (0, 1)\cup (1, \infty)$. Then, 
\begin{align}
I_{\alpha}^{\text{S}}(X; Y) 
&= \max_{r_{X\mid Y}} F_{\alpha}^{\text{S1}}(p_{X}, r_{X\mid Y})= \max_{r_{X\mid Y}} F_{\alpha}^{\text{S2}}(p_{X}, r_{X\mid Y}), \\
I_{\alpha}^{\text{A}}(X; Y) 
&= \max_{r_{X\mid Y}} F_{\alpha}^{\text{A1}}(p_{X}, r_{X\mid Y})= \max_{r_{X\mid Y}} F_{\alpha}^{\text{A2}}(p_{X}, r_{X\mid Y}), 
\end{align}
where 
$F_{\alpha}^{\text{S1}}(p_{X}, r_{X\mid Y}) := \frac{\alpha}{\alpha-1}\log \sum_{x, y}p_{X}(x)^{\alpha}p_{Y\mid X}(y | x)$ $r_{X\mid Y}(x | y)^{1-\frac{1}{\alpha}}$, 
$F_{\alpha}^{\text{S2}}(p_{X}, r_{X\mid Y}) := F_{\alpha}^{\text{S1}}(p_{X}, r_{X_{\alpha}\mid Y})$,
$F_{\alpha}^{\text{A1}}(p_{X}, r_{X\mid Y}) := F_{\alpha}^{\text{S1}}(p_{X_{\alpha}}, r_{X\mid Y})$, 
$F_{\alpha}^{\text{A2}}(p_{X}, r_{X\mid Y}) := F_{\alpha}^{\text{S1}}(p_{X_{\alpha}}, r_{X_{\alpha}\mid Y})$, 
and $p_{X_{\alpha}}$ and $r_{X_{\alpha}\mid Y}$ are the $\alpha$-tilted distributions of $p_{X}$ and $r_{X\mid Y}$ 
defined as \eqref{eq:alpha_tilted_dist} and $r_{X_{\alpha}\mid Y}(x | y) := \frac{r_{X\mid Y}(x | y)^{\alpha}}{\sum_{x}r_{X\mid Y}(x | y)^{\alpha}}$. 
\end{prop}

From Proposition \ref{prop:vc_Sibson_Arimoto_MI}, the Sibson and Arimoto capacity can be represented as double maximization problems in the form of   
$C_{\alpha}^{(\cdot)} = \max_{p_{X}} \max_{r_{X\mid Y}} F_{\alpha}^{(\cdot)}(p_{X}, r_{X\mid Y})$. 
Arimoto proposed an AO algorithm for computing $C_{\alpha}^{\text{S}}$ (Algorithm \ref{alg:Arimoto}, referred to as \textit{Arimoto algorithm}) based on the representation of $C_{\alpha}^{\text{S}} = \max_{p_{X}} \max_{r_{X\mid Y}} F_{\alpha}^{S1}(p_{X}, r_{X\mid Y})$, where $p_{X}^{(0)}$ is the initial distribution \cite{arimoto1977,1055640}. 
By replacing $(p_{X}, r_{X\mid Y})$ in Algorithm \ref{alg:Arimoto} with $(p_{X_{\alpha}}, r_{X\mid Y}), (p_{X}, r_{X_{\alpha}\mid Y})$, or $(p_{X_{\alpha}}, r_{X_{\alpha}\mid Y})$, 
the other algorithms corresponding to each representation can be obtained. Note that these AO algorithms are equivalent if the initial distributions of each algorithm are appropriately selected \cite[Cor 1]{kamatsuka2024new}. 

\begin{algorithm}[h]
	\caption{Arimoto algorithm for computing $C_{\alpha}^{\text{S}}$ \cite{arimoto1977,1055640}}
	\label{alg:Arimoto}
	\begin{algorithmic}[1]
		\Require 
			\Statex $p_{X}^{(0)}, p_{Y\mid X}$, $\epsilon > 0$
		\Ensure
			\Statex approximate value of $C_{\alpha}^{\text{S}}$
		\Initialize{
			$r_{X\mid Y}^{(0)}(x|y) \gets \frac{p^{(0)}_{X}(x)p_{Y\mid X}(y\mid x)^{\alpha}}{\sum_{x}p^{(0)}_{X}(x)p_{Y\mid X}(y\mid x)^{\alpha}}$\\ 
			$F^{(0)}\gets F_{\alpha}^{\text{S1}}(p_{X}^{(0)}, r_{X\mid Y}^{(0)})$ \\ 
			$k\gets 0$ \\
      }
		\Repeat
			\State $p_{X}^{(k+1)}(x) \gets \frac{\left(\sum_{y}p_{Y\mid X}(y\mid x)r^{(k)}_{X\mid Y}(x\mid y)^{1-\frac{1}{\alpha}} \right)^{\frac{\alpha}{\alpha-1}}}{\sum_{x} \left( \sum_{y}p_{Y\mid X}(y\mid x)r^{(k)}_{X\mid Y}(x\mid y)^{1-\frac{1}{\alpha}} \right)^{\frac{\alpha}{\alpha-1}}}$
			\State $k\gets k+1$
			\State $r_{X\mid Y}^{(k)}(x|y) \gets \frac{p^{(k)}_{X}(x)p_{Y\mid X}(y\mid x)^{\alpha}}{\sum_{x}p^{(k)}_{X}(x)p_{Y\mid X}(y\mid x)^{\alpha}}$
			\State $F^{(k)} \gets F_{\alpha}^{\text{S1}}(p_{X}^{(k)}, r_{X\mid Y}^{(k)})$
		\Until{$\abs{F^{(k)} - F^{(k-1)}} < \epsilon$} 
		\State \textbf{return} $F^{(k)}$
	\end{algorithmic}
\end{algorithm}

\subsection{Jitsumatsu--Oohama Algorithm for Computing the Sibson Capacity}
Jitsumatsu and Oohama \cite{8889422} provided an AO algorithm for computing $\min_{p_{X}}E_{0}(\rho, p_{X})$ based on the following variational characterization for $\rho\in (-1, 0)$:\footnote{This characterization can be extended to where a cost constraint exists \cite{8889422}.}  
\begin{prop}[\text{\cite[Prop 1 and Lemma 2]{8889422}}]
Let $\rho\in (-1, 0)$. Then, 
\begin{align}
\min_{p_{X}}E_{0}(\rho, p_{X}) &= \min_{q_{X, Y}}\min_{\tilde{q}_{X, Y}} F_{\rho}^{\text{JO}}(q_{X, Y}, \tilde{q}_{X, Y}), \label{eq:vc_min_Gallager} 
\end{align}
where 
\begin{align}
&F_{\rho}^{\text{JO}}(q_{X, Y}, \tilde{q}_{X, Y}) \notag \\
& := \vE_{X, Y}^{q_{X, Y}} \left[\log \frac{\tilde{q}_{Y\mid X}(Y\mid X)^{1+\rho}\tilde{q}_{Y}(Y)^{-\rho}}{p_{Y\mid X}(Y\mid X)}\right] + D(q_{X, Y} || \tilde{q}_{X, Y}), 
\end{align}
$\tilde{q}_{Y}(y):=\sum_{x}\tilde{q}_{X, Y}(x, y)$, and $\tilde{q}_{Y\mid X}(y | x):=\frac{\tilde{q}_{X, Y}(x, y)}{\sum_{y}\tilde{q}_{X, Y}(x, y)}$. 
\end{prop}

From this characterization in \eqref{eq:vc_min_Gallager} and the closed-form expression of the Sibson MI in \eqref{eq:closed_form_Sibson_MI}, the following variational characterization of the Sibson capacity 
$C_{\alpha}^{\text{S}}$ for $\alpha\in (1, \infty)$ can be obtained.
\begin{cor}
Let $\alpha\in (1, \infty)$. Then, 
\begin{align}
C_{\alpha}^{\text{S}} &= \max_{q_{X, Y}}\max_{\tilde{q}_{X, Y}} \tilde{F}_{\alpha}^{\text{JO}}(q_{X, Y}, \tilde{q}_{X, Y}), \label{eq:vc_Sibson_capacity_JO}
\end{align}
where $\tilde{F}_{\alpha}^{\text{JO}}(q_{X, Y}, \tilde{q}_{X, Y}) := \frac{\alpha}{1-\alpha}F_{\frac{1}{\alpha}-1}^{\text{JO}}(q_{X, Y}, \tilde{q}_{X, Y})$. 
\end{cor}
Based on the characterization in \eqref{eq:vc_Sibson_capacity_JO}, another AO algorithm, referred to as \textit{Jitsumasu--Oohama algorithm}, for computing $C_{\alpha}^{\text{S}}$ is obtained as described in Algorithm \ref{alg:JOA}, 
where $q_{X, Y}^{(0)}$ is the initial distribution. 

\begin{algorithm}[h]
	\caption{Jitsumatsu--Oohama algorithm for computing $C_{\alpha}^{\text{S}}, \alpha\in (1, \infty)$ \cite{8889422}}
	\label{alg:JOA}
	\begin{algorithmic}[1]
		\Require 
			\Statex $q_{X, Y}^{(0)}, p_{Y\mid X}$, $\epsilon > 0$
		\Ensure
			\Statex approximate value of $C_{\alpha}^{\text{S}}$
		\Initialize{
			$\tilde{q}_{X, Y}^{(0)}(x, y) \gets q_{X, Y}^{(0)}(x, y)$
			$F^{(0)}\gets \tilde{F}_{\alpha}^{\text{JO}}(q_{X, Y}^{(0)}, \tilde{q}_{X, Y}^{(0)})$ \\ 
			$k\gets 0$ \\
      }
		\Repeat
			\State $q_{X, Y}^{(k+1)}(x, y) \gets \frac{\tilde{q}_{X \mid Y}^{(k)}(x|y)^{1-\frac{1}{\alpha}}p_{Y\mid X}(y|x)\tilde{q}_{X}^{(k)}(x)^{\frac{1}{\alpha}}}{\sum_{x, y}\tilde{q}_{X \mid Y}^{(k)}(x|y)^{1-\frac{1}{\alpha}}p_{Y\mid X}(y|x)\tilde{q}_{X}^{(k)}(x)^{\frac{1}{\alpha}}}$
			\State $k\gets k+1$
			\State $\tilde{q}_{X, Y}^{(k)}(x, y) \gets q_{X, Y}^{(k)}(x, y)$
			\State $F^{(k)}\gets \tilde{F}_{\alpha}^{\text{JO}}(q_{X, Y}^{(k)}, \tilde{q}_{X, Y}^{(k)})$
		\Until{$\abs{F^{(k)} - F^{(k-1)}} < \epsilon$} 
		\State \textbf{return} $F^{(k)}$
	\end{algorithmic}
\end{algorithm}

\section{Proposed Algorithm}\label{sec:main_result}
Here, we propose an AO algorithm for computing the Augustin--Csisz{\'a}r capacity $C_{\alpha}^{\text{C}}:=\max_{p_{X}}I_{\alpha}^{\text{C}}(X; Y)$ 
for $\alpha\in(1, \infty)$ based on the following variational characterization of the Augustin--Csisz{\'a}r MI, which is recently proposed by Kamatsuka \textit{et al}. \cite{kamatsuka2024algorithms}.

\begin{prop}[\text{\cite[Prop 4]{kamatsuka2024algorithms}}] \label{prop:vc_Csiszar_MI}
Let $\alpha\in(1, \infty)$. Then, 
\begin{align}
I_{\alpha}^{\text{C}}(p_{X}, p_{Y\mid X}) &= \max_{\tilde{q}_{Y\mid X}}\max_{r_{X\mid Y}} \tilde{F}_{\alpha}^{\text{C}}(p_{X}, \tilde{q}_{Y\mid X}, r_{X\mid Y}), 
\end{align}
where 
\begin{align}
\tilde{F}_{\alpha}^{\text{C}}(p_{X}, \tilde{q}_{Y\mid X}, r_{X\mid Y}) &:= \vE_{X, Y}^{p_{X}\tilde{q}_{Y\mid X}} \left[\log \frac{r_{X\mid Y}(X\mid Y)}{p_{X}(X)}\right] \notag \\ 
& \qquad  + \frac{\alpha}{1-\alpha} D(p_{X}\tilde{q}_{Y\mid X} || p_{X}p_{Y\mid X}).
\end{align}
\end{prop}
Thus, the Augustin--Csisz{\'a}r capacity $C_{\alpha}^{\text{C}}$ can be represented as the following triple maximization problem: 
\begin{align}
C_{\alpha}^{\text{C}} &= \max_{p_{X}}\max_{\tilde{q}_{Y\mid X}}\max_{r_{X\mid Y}} \tilde{F}_{\alpha}^{\text{C}}(p_{X}, \tilde{q}_{Y\mid X}, r_{X\mid Y}). \label{eq:triple_max}
\end{align} 

Now, we consider solving the optimization problem in \eqref{eq:triple_max}. 
To this end, we derive the updating formulae for the AO algorithm as follows:

\begin{prop} \label{prop:update_formulae_Csiszar_capacity}
Let $\alpha\in (1, \infty)$. Then, the following holds:
\begin{enumerate}
\item For a fixed $(p_{X}, \tilde{q}_{Y\mid X})$, $\tilde{F}_{\alpha}^{\text{C}}(p_{X}, \tilde{q}_{Y\mid X}, r_{X\mid Y})$ is maximized by 
\begin{align}
r_{X\mid Y}^{*}(x\mid y) &:= \frac{p_{X}(x)\tilde{q}_{Y\mid X}(y\mid x)}{\sum_{x} p_{X}(x)\tilde{q}_{Y\mid X}(y\mid x)}. \label{eq:update_r_X_Y}
\end{align}
\item For a fixed $(p_{X}, r_{X\mid Y})$, $\tilde{F}_{\alpha}^{\text{C}}(p_{X}, \tilde{q}_{Y\mid X}, r_{X\mid Y})$ is maximized by 
\begin{align}
\tilde{q}_{Y\mid X}^{*}(y\mid x) &:= \frac{p_{Y\mid X}(y \mid x) r_{X\mid Y}(x \mid y)^{1-\frac{1}{\alpha}}}{\sum_{y} p_{Y\mid X}(y \mid x) r_{X\mid Y}(x \mid y)^{1-\frac{1}{\alpha}}}. \label{eq:update_q_Y_X}
\end{align}
\item For a fixed $(\tilde{q}_{Y\mid X}, r_{X\mid Y})$, $\tilde{F}_{\alpha}^{\text{C}}(p_{X}, \tilde{q}_{Y\mid X}, r_{X\mid Y})$ is maximized by 
\begin{align}
&p_{X}^{*}(x) := \frac{g_{\tilde{q}_{Y\mid X}, r_{X\mid Y}}(x)}{\sum_{x}g_{\tilde{q}_{Y\mid X}, r_{X\mid Y}}(x)}, \label{eq:update_p_X} 
\end{align}
where 
\begin{align}
&g_{\tilde{q}_{Y\mid X}, r_{X\mid Y}}(x) \notag \\ 
&:=\exp\Big\{\sum_{y}\tilde{q}_{Y\mid X}(y\mid x)\log \tilde{q}_{Y\mid X}(y\mid x)^{\frac{\alpha}{1-\alpha}} \notag \\ 
&\qquad \qquad \qquad \qquad \times r_{X\mid Y}(x\mid y)p_{Y\mid X}(y\mid x)^{\frac{\alpha}{\alpha-1}}\Big\}.
\end{align}
\end{enumerate}
\end{prop}
\begin{proof}
Proofs of \eqref{eq:update_r_X_Y} and \eqref{eq:update_q_Y_X} are presented in \cite[Prop 5]{kamatsuka2024algorithms}. 
Thus, we only prove \eqref{eq:update_p_X}. 
Note that $\tilde{F}_{\alpha}^{\text{C}}(p_{X}, \tilde{q}_{Y\mid X}, r_{X\mid Y})$ is concave with respect to $p_{X}$ for a fixed $(\tilde{q}_{Y\mid X}, r_{X\mid Y})$ since $\vE_{X, Y}^{p_{X}\tilde{q}_{Y\mid X}} \left[\log \frac{r_{X\mid Y}(X\mid Y)}{p_{X}(X)}\right]$ is concave with respect to $p_{X}$ \cite[Sec 10.3.2]{10.5555/1199866}, $D(p_{X}\tilde{q}_{Y\mid X} || p_{X}p_{Y\mid X})$ is convex with respect to $p_{X}$ \cite[Thm 2.7.2]{Cover:2006:EIT:1146355}, and $\frac{\alpha}{1-\alpha} < 0$ for $\alpha\in (1, \infty)$.  
Here, define a Lagrangian $J(p_{X})$ as 
\begin{align}
J(p_{X}) &:= \tilde{F}_{\alpha}^{\text{C}}(p_{X}, \tilde{q}_{Y\mid X}, r_{X\mid Y}) + \lambda (\sum_{x}p_{X}(x)-1), 
\end{align}
where $\lambda$ is a Lagrange multiplier. By solving the equation $\frac{\partial J}{\partial p_{X}(x)}=0$ with respect to $p_{X}(x)$, we obtain \eqref{eq:update_p_X}. 
\end{proof}
Thus, the AO algorithm for computing $C_{\alpha}^{\text{C}}$ is obtained as described in Algorithm \ref{alg:Csiszar}, where $(p_{X}^{(0)}, \tilde{q}_{Y\mid X}^{(0)})$  is the initial distribution. 

\begin{algorithm}[h]
	\caption{AO algorithm for computing $C_{\alpha}^{\text{C}}, \alpha\in (1, \infty)$}
	\label{alg:Csiszar}
	\begin{algorithmic}[1]
		\Require 
			\Statex $p_{X}^{(0)}, \tilde{q}_{Y\mid X}^{(0)}, p_{Y\mid X}$, $\epsilon > 0$
		\Ensure
			\Statex approximate value of $C_{\alpha}^{\text{C}}$
		\Initialize{
			$r_{X\mid Y}^{(0)}(x | y) \gets \frac{p_{X}^{(0)}(x)\tilde{q}_{Y\mid X}^{(0)}(y\mid x)}{\sum_{x} p_{X}^{(0)}(x)\tilde{q}_{Y\mid X}^{(0)}(y\mid x)}$
			$F^{(0)}\gets \tilde{F}_{\alpha}^{\text{C}}(p_{X}^{(0)}, \tilde{q}_{Y\mid X}^{(0)}, r_{X\mid Y}^{(0)})$ \\ 
			$k\gets 0$ \\
      }
		\Repeat
			\State $p_{X}^{(k+1)}(x) \gets \frac{g_{\tilde{q}^{(k)}_{Y\mid X}, r^{(k)}_{X\mid Y}}(x)}{\sum_{x}g_{\tilde{q}^{(k)}_{Y\mid X}, r^{(k)}_{X\mid Y}}(x)}$
			\State $\tilde{q}_{Y\mid X}^{(k+1)}(y\mid x) \gets \frac{p_{Y\mid X}(y \mid x) r_{X\mid Y}^{(k)}(x \mid y)^{1-\frac{1}{\alpha}}}{\sum_{y} p_{Y\mid X}(y \mid x) r_{X\mid Y}^{(k)}(x \mid y)^{1-\frac{1}{\alpha}}}$
			\State $k\gets k+1$
			\State $r_{X\mid Y}^{(k)}(x | y) \gets \frac{p_{X}^{(k)}(x)\tilde{q}_{Y\mid X}^{(k)}(y\mid x)}{\sum_{x} p_{X}^{(k)}(x)\tilde{q}_{Y\mid X}^{(k)}(y\mid x)}$
			\State $F^{(k)}\gets \tilde{F}_{\alpha}^{\text{C}}(p_{X}^{(k)}, \tilde{q}_{Y\mid X}^{(k)}, r_{X\mid Y}^{(k)})$
		\Until{$\abs{F^{(k)} - F^{(k-1)}} < \epsilon$} 
		\State \textbf{return} $F^{(k)}$
	\end{algorithmic}
\end{algorithm}

So far, three different algorithms (Algorithm \ref{alg:Arimoto}, \ref{alg:JOA}, and \ref{alg:Csiszar}) for calculating $\alpha$-capacity $C_{\alpha}^{\text{S}}=C_{\alpha}^{\text{A}}=C_{\alpha}^{\text{C}}, \alpha \in (1, \infty)$ have been obtained. 
For each algorithm, the required computational complexity to compute the updating formulae at each iteration is $O(\abs{\mathcal{X}}\cdot\abs{\mathcal{Y}})$.
In the next section, we compare the speed of convergence of these algorithms through numerical examples.

\begin{table*}[htb]
  \caption{Approximate values of $\alpha$-capacity and the number of iterations $(F^{(N)}, N)$}
  \label{tab:alpha_capacity}
  \centering
  \begin{tabular}{@{} lllll@{}}
    \toprule
    \thead{Algorithm}        & \thead{$\alpha=1.03$} & \thead{$\alpha=1.5$} & \thead{$\alpha=2.0$} & \thead{$\alpha=5.0$} \\ 
    \midrule
    Algorithm \ref{alg:Arimoto} with $p_{X}^{(0)} = u_{X}$
    & $(0.054204678, 851)$ & $(0.07617995, 738)$ & $(0.097030615, 790)$ & $(0.183426237, 20)$ \\ 
	Algorithm \ref{alg:JOA} with $q_{X, Y}^{(0)}= u_{X, Y}$ 
	& $(0.054201694, 19076)$ & $(0.07617965, 2744)$ & $(0.097030139,2258)$ & $(0.183426230,113)$ \\
	Algorithm \ref{alg:JOA} with $q_{X, Y}^{(0)}= u_{X}\times p_{Y\mid X}$ 
	& $(0.054201694, 19075)$ & $(0.07617965, 2743)$ & $(0.097030139, 2257)$ & $(0.183426230, 112)$ \\ 
	Algorithm \ref{alg:Csiszar} with $p_{X}^{(0)}\tilde{q}_{Y\mid X}^{(0)}= u_{X, Y}$ 
	& $(0.054204678, 912)$ & $(0.07617991, 1059)$ & $(0.097030445, 1347)$ & $(0.183426230, 97)$ \\ 
	Algorithm \ref{alg:Csiszar} with $p_{X}^{(0)}\tilde{q}_{Y\mid X}^{(0)} = u_{X}\times p_{Y\mid X}$ 
	& $(0.054204678, 874)$ & $(0.07617991, 1055)$ & $(0.097030445, 1344)$ & $(0.183426231, 95)$ \\ 
    \bottomrule
  \end{tabular}
\end{table*}

\begin{figure*}[t]
\centering
\begin{tabular}{@{}cc@{}}
\includegraphics[width=8cm, clip]{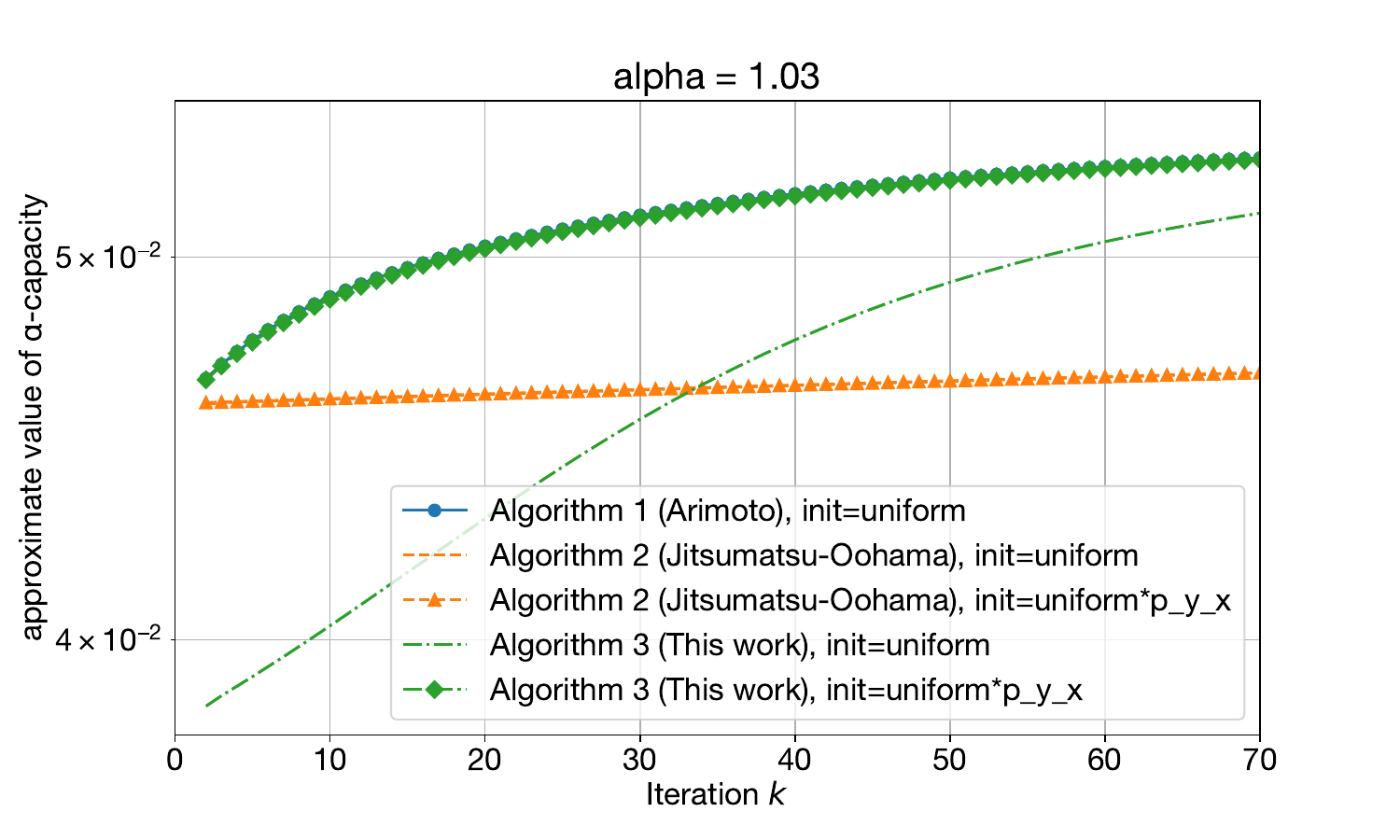} & \includegraphics[width=8cm, clip]{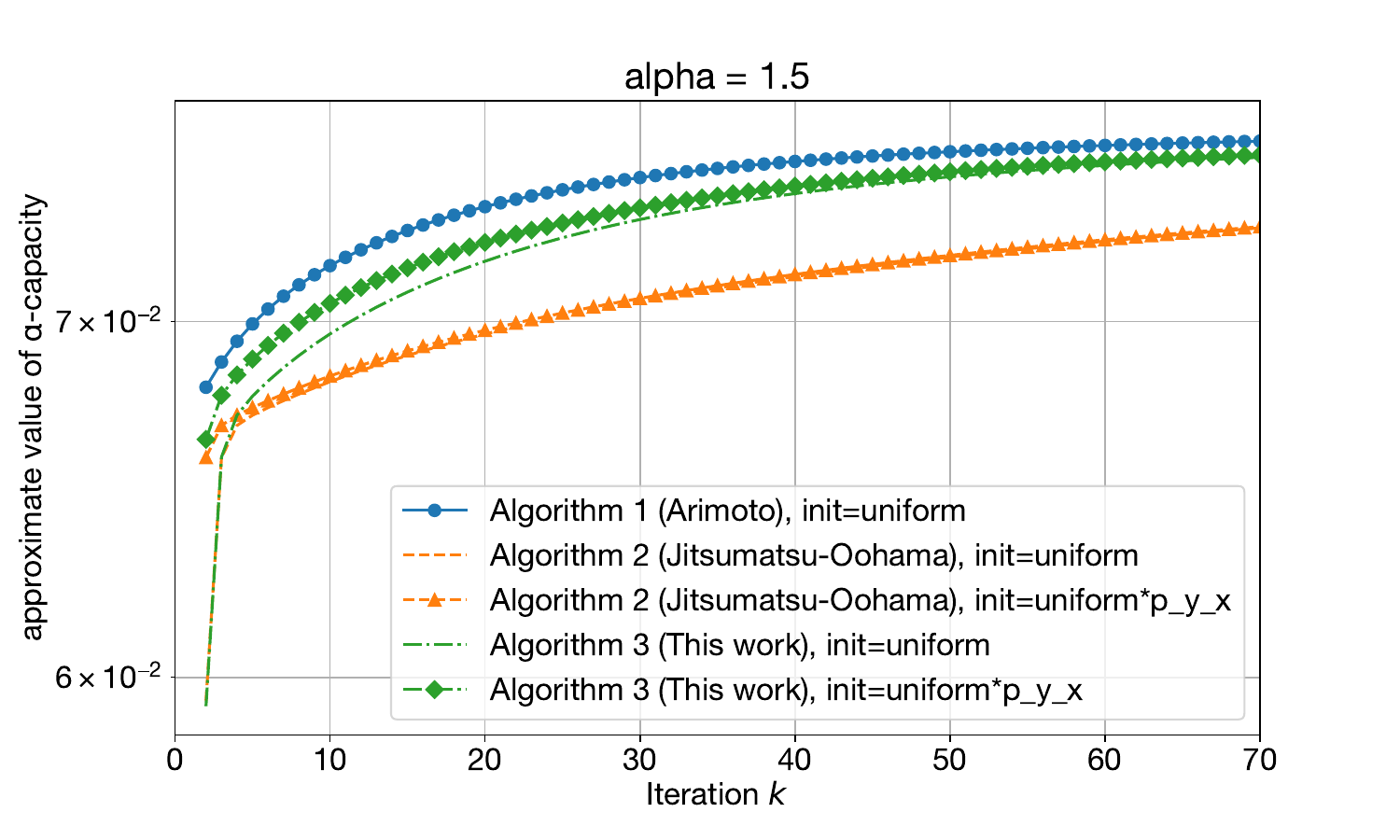} \\ 
(a) $\alpha=1.03$ & (b) $\alpha=1.5$ \\ 
\includegraphics[width=8cm, clip]{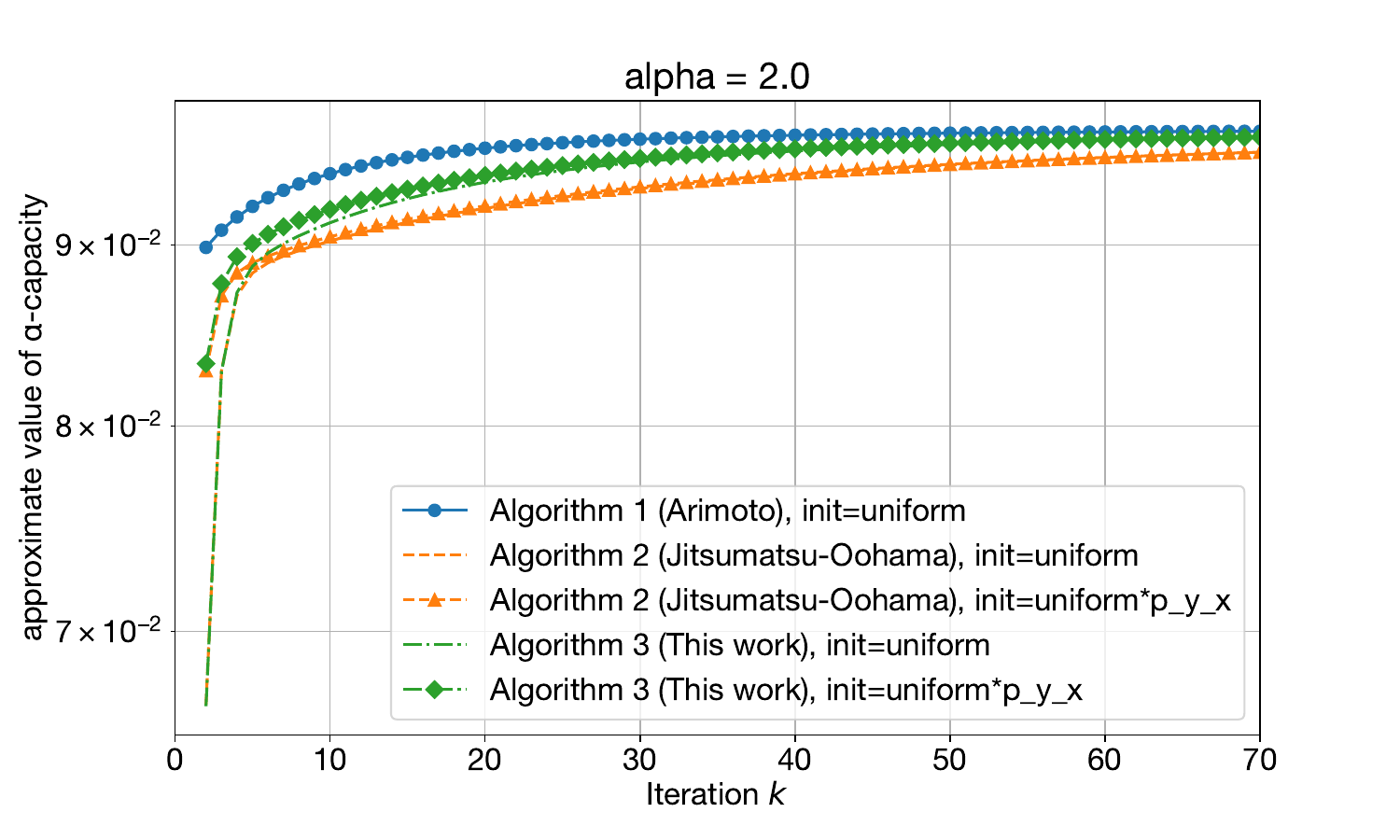} & \includegraphics[width=8cm, clip]{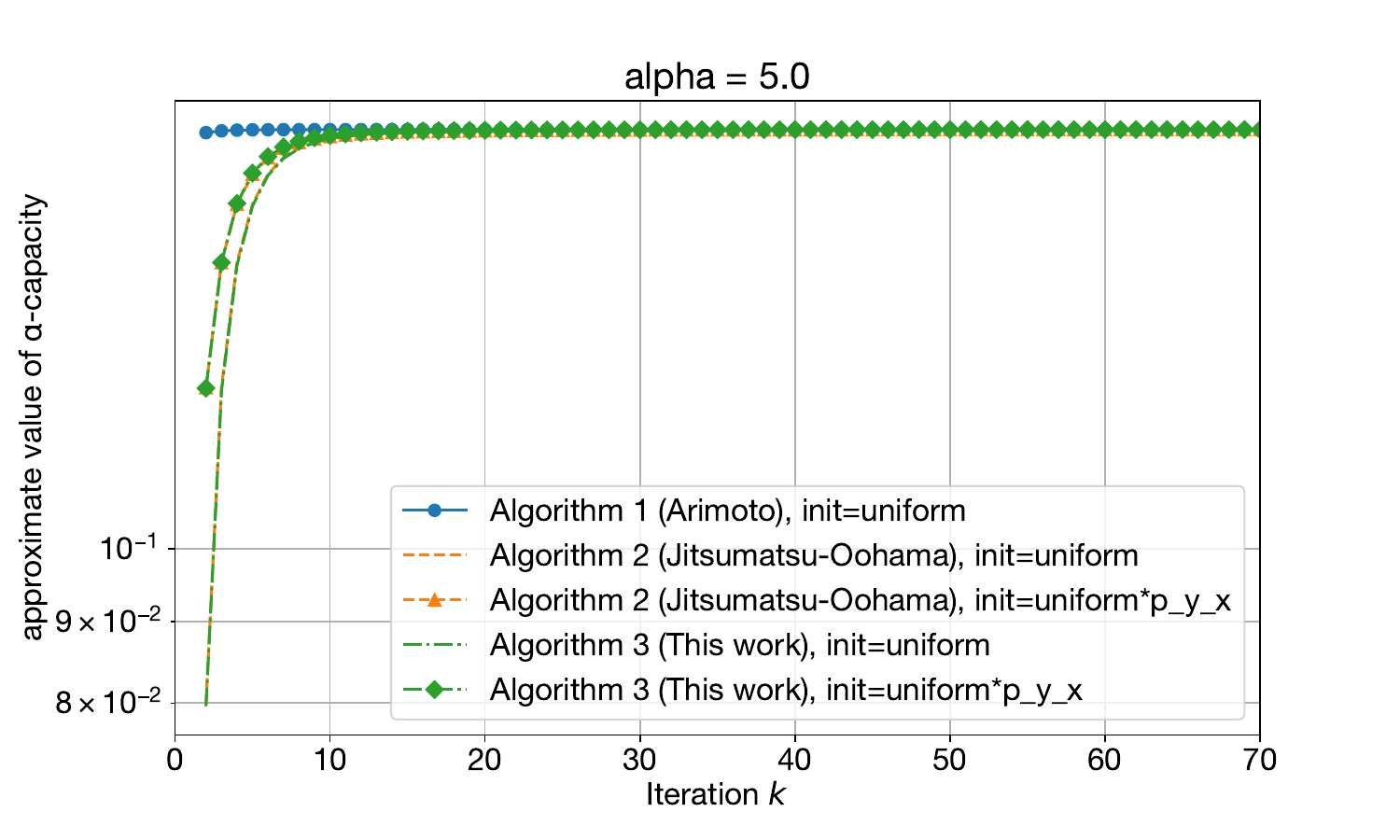} \\  
(c) $\alpha=2.0$  & (d) $\alpha=5.0$
\end{tabular}
\caption{Transitions of approximate value of $\alpha$-capacity $F^{(k)}$ as $k$ increases for (a) $\alpha=1.03$, (b) $\alpha=1.5$, (c) $\alpha=2.0$, and (d) $\alpha=5.0$. Blue solid curve expresses Algorithm \ref{alg:Arimoto} (Arimoto algorithm), orange dashed curve expresses Algorithm \ref{alg:JOA} (Jitsumasu--Oohama algorithm), and green dash-dotted curve expresses Algorithm \ref{alg:Csiszar} (This work), respectively.}
\label{fig:numerical_example}
\end{figure*}

\section{Numerical Example} \label{sec:numerical_example}
This section presents numerical examples to compare the convergence speeds of Algorithm \ref{alg:Arimoto}--\ref{alg:Csiszar} for computing $\alpha$-capacity 
$C_{\alpha}^{\text{S}}=C_{\alpha}^{\text{A}}=C_{\alpha}^{\text{C}}, \alpha \in (1, \infty)$ of the following channel matirix $p_{Y\mid X}$:
\begin{align}
p_{Y\mid X} &:= 
\begin{bmatrix}
0.259 & 0.463 & 0.278 \\ 
0.328 & 0.172 & 0.500 \\ 
0.425 & 0.225 & 0.350
\end{bmatrix}, 
\end{align}
where $(i, j)$-element of the channel matrix\footnote{The elements of the channel matrix $p_{Y\mid X}$ are randomly generated from a uniform distribution. 
Note that numerical examples were computed using Algorithms \ref{alg:Arimoto}--\ref{alg:Csiszar} for other randomly generated channel matrices that did not appear in this study, wherein the behavior of the algorithms is generally the same.} corresponds to the conditional probability $p_{Y\mid X}(j | i), i\in \mathcal{X}=\{1, 2, 3\}, j\in \mathcal{Y}=\{1, 2, 3\}$.
Table \ref{tab:alpha_capacity} shows the number of iteration $N$ and approximate values of $\alpha$-capacity  of the channel $p_{Y\mid X}$, denoted as $F^{(N)}$,  for $\alpha\in \{1.03, 1.5, 2.0, 5.0\}$ computed by each algorithm with specific initial conditions, where  
\begin{itemize}
\item  $N$ is the number of iterations needed to satisfy the stopping conditions of the algorithms when $\epsilon=10^{-9}$ is set. 
\item $u_{X} (x) := \frac{1}{\abs{\mathcal{X}}}$. (uniform distribution on $\mathcal{X}$)
\item $u_{X, Y} (x, y) := \frac{1}{\abs{\mathcal{X}}\cdot\abs{\mathcal{Y}}}$. (uniform distribution on $\mathcal{X}\times \mathcal{Y}$)
\end{itemize}

Figure \ref{fig:numerical_example} illustrates transitions of the approximate values of $\alpha$-capacity $F^{(k)}$ as $k$ increases computed by each algorithm with specific initial conditions.
From Table \ref{tab:alpha_capacity} and Figure \ref{fig:numerical_example}, we can observe that 
\begin{itemize}
\item The number of iterations $N$ needed for convergence decreases as $\alpha$ increases for all algorithms. 
\item Algorithm \ref{alg:Arimoto} shows the fastest convergence speed among algorithms, followed by Algorithm \ref{alg:Csiszar} and Algorithm \ref{alg:JOA}.
The convergence speed of Algorithm \ref{alg:JOA} is significantly slow when $\alpha$ is close to $1.0$. 
\item The effect of setting $p_{Y\mid X}$ to the initial distributions is larger for smaller $\alpha$ in the case of Algorithm \ref{alg:Csiszar}. 
Specifically, when $\alpha$ is close to $1.0$, the convergence performance of Algorithm \ref{alg:Csiszar} with $p_{X}^{(0)}\tilde{q}_{Y\mid X}^{(0)} = u_{X}\times p_{Y\mid X}$ is comparable to that of Algorithm \ref{alg:Arimoto}. 
In contrast, the effect seems negligible  in the case of Algorithm \ref{alg:JOA}.
\end{itemize}

\section{Conclusion}\label{sec:conclusion}
In this paper, we proposed the AO algorithm for directly computing Augustin--Csisz{\'a}r capacity $C_{\alpha}^{\text{C}}$ for $\alpha \in (1, \infty)$
based on the variational characterization of Augustin--Csisz{\'a}r mutual information proposed by Kamatsuka \textit{et al.} \cite{kamatsuka2024algorithms}. 
The proposed algorithm can also compute the correct decoding exponent for a fixed $\rho \in (-1, 0)$ by modifying its objective function. 
Furthermore, we compared the convergence speed with previously proposed algorithms for computing $\alpha$-capacity through numerical examples. 
In a future study, we will prove the global convergence for the Algorithm \ref{alg:Csiszar} and provide a theoretical analysis of the convergence speed of these algorithms.


\end{document}